 \documentclass[aps,showpacs,superscriptaddress,nofootinbib,reprint]{revtex4-2} 

\usepackage{subcaption}

\usepackage{graphicx}
\usepackage[toc,page]{appendix}
\usepackage{dcolumn}
\usepackage{braket}
\usepackage{amsmath,amsfonts,amssymb,amstext,amscd,amsthm,bbold}
\usepackage{comment,float}

\usepackage[dvipsnames]{xcolor}
\usepackage[colorlinks=true,linkcolor=blue,urlcolor=magenta,citecolor=magenta]{hyperref}
\usepackage[normalem]{ulem}
\usepackage{lipsum}
\usepackage{dcolumn}
\usepackage{bm}
\usepackage{verbatim}

\usepackage{orcidlink}
\hypersetup{
	colorlinks=true,
	citecolor= BrickRed,
	linkcolor=BrickRed,
	filecolor=BrickRed,      
	urlcolor=BrickRed,
	pdftitle={Overleaf Example},
	pdfpagemode=FullScreen,
}

\usepackage{titlesec}
\titleformat{\section}
{\sffamily\bfseries} 
{\thesection} 
{1em} 
{} 

\titleformat{\subsection}
{\sffamily\bfseries} 
{\thesubsection} 
{1em} 
{} 

\usepackage[T1]{fontenc}
\usepackage[utf8]{inputenc}
\usepackage{enumitem}

\titlespacing\section{0pt}{12pt plus 4pt minus 2pt}{2pt plus 2pt minus 2pt}
\titlespacing\subsection{0pt}{12pt plus 4pt minus 2pt}{2pt plus 2pt minus 2pt}
\titlespacing\subsubsection{0pt}{12pt plus 4pt minus 2pt}{2pt plus 2pt minus 2pt}
\usepackage{lipsum} 

\newtheorem*{rep@theorem}{\rep@title}
\newcommand{\newreptheorem}[2]{%
\newenvironment{rep#1}[1]{%
 \def\rep@title{#2 \ref{##1}}%
 \begin{rep@theorem}}%
 {\end{rep@theorem}}}
\makeatother

\newtheorem{theorem}{Theorem}
\newreptheorem{theorem}{Theorem}
\newtheorem{lemma}{Lemma}

\usepackage{bbm}

\usepackage{verbatim}
\usepackage{amssymb}

\newcommand{\be}{\begin{equation}}
\newcommand{\ee}{\end{equation}}
\newcommand{\bea}{\begin{eqnarray}}
\newcommand{\eea}{\end{eqnarray}}

\def\squareforqed{\hbox{\rlap{$\sqcap$}$\sqcup$}}
\def\qed{\ifmmode\squareforqed\else{\unskip\nobreak\hfil
\penalty50\hskip1em\null\nobreak\hfil\squareforqed
\parfillskip=0pt\finalhyphendemerits=0\endgraf}\fi}
\def\endenv{\ifmmode\;\else{\unskip\nobreak\hfil
\penalty50\hskip1em\null\nobreak\hfil\;
\parfillskip=0pt\finalhyphendemerits=0\endgraf}\fi}

\newcommand{\tr}{\text{tr}}
\newcommand{\I}{\mathbbm{1}}
\newcommand{\im}{\mathbbm{i}}

\newcommand{\blk}{\color{black}}

\makeatletter
\let\@afterindenttrue\@afterindentfalse
\makeatother

\begin{document}


\title{Single system based generation of certified randomness using Leggett-Garg inequality
}

\author{Pingal Pratyush Nath}
\affiliation{Indian Institute of Science,  C. V. Raman Road, Bengaluru, Karnataka 560012, India}

\author{Debashis Saha}
\affiliation{School of Physics, Indian Institute of Science Education and Research Thiruvananthapuram, Thiruvananthapuram, Kerala 695551, India}

\author{Dipankar Home}
\affiliation{Center for Astroparticle Physics and Space Science (CAPSS),Bose Institute, Kolkata 700 091, India.}

\author{Urbasi Sinha }
\email[]{usinha@rri.res.in}
\affiliation{Raman Research Institute, C. V. Raman Avenue, Sadashivanagar, Bengaluru, Karnataka 560080, India}

\date{\today}

\begin{abstract}
    We theoretically formulate and experimentally demonstrate a secure scheme for semi-device-independent quantum random number generation by utilizing Leggett-Garg inequality violations, within a loophole-free photonic architecture. The quantification of the generated randomness is rigorously estimated by  analytical as well as numerical approaches, both of which are in perfect agreement. We securely generate $919,118$ truly unpredictable bits at a rate of $3865$ bits/sec. This opens up an unexplored avenue towards an empirically convenient class of reliable random number generators harnessing the quantumness of single systems.
    
\end{abstract}


\maketitle

\textit{Introduction : }The production and characterization of true random numbers as a resource for various applications is currently a cutting-edge topic attracting considerable studies. In particular, the encryption schemes used in all protocols for secure communication, including
quantum cryptography, rely on genuinely unpredictable random numbers. This is necessary to ensure that an adversary cannot decipher the encrypted message. Furthermore, the desired security must be guaranteed even in the presence of device imperfections or any tampering by an adversary. Strikingly, these key requirements for ensuring reliable private randomness are not currently satisfied by any random number generator (RNG).\cite{herrero2017quantum,ma2016quantum, mannalatha2023comprehensive, markowsky2014sad}

On the other hand, studies over the last decade have opened up an avenue for developing fully secure device-independent RNGs Table I based on using quantum entangled states and certifying genuine randomness by using quantum non-locality evidenced through the statistical violation of Bell inequality \cite{colbeck2009quantum,pironio2010random,colbeck2011private, acin2012randomness,colbeck2012free, pironio2013security, abellan2015generation,acin2016certified,bierhorst2018experimentally, liu2018high,liu2018device, shalm2021device}.          
But an empirical impediment in realizing practically viable such device-independent RNGs is the requirement of adequate spatial separation between two parties while making the Bell inequality testing measurements on their joint state by preserving their entanglement across distance\cite{pironio2018certainty}. To obviate this difficulty, we provide in this paper a proof-of-concept demonstration of how the quantumness of an individual system, as evidenced through the observable violation of the temporal counterpart of Bell inequality\cite{bell1964einstein, clauser1969proposed,brunner2014bell}, viz., the Leggett-Garg inequality(LGI), can be harnessed to certify and quantify genuine randomness.

Ever since LGI  was formulated \cite{leggett1985quantum,emary2013leggett} as a consequence of the assumptions characterizing
the  notion of macrorealism, studies related to LGI have largely focused on using LGI for
testing and probing ramifications of the quantum mechanical (QM) violation of macrorealism \cite{palacios2010experimental,athalye2011investigation,dressel2011experimental,katiyar2013violation,athalye2011investigation,emary2012leggett,williams2008weak,knee2012violation,formaggio2016violation,xu2011experimental,dressel2011experimental,goggin2011violation,suzuki2012violation,wang2018violations, katiyar2017experimental, majidy2021detecting,katiyar2017experimental,ku2020experimental}.
On the other hand, in the present work, we focus on a specific applicational feature of LGI. Apart from being derivable from macrorealism, LGI can also be derived from the conjunction of the assumptions of perfect predictability and No-Signaling-in-Time (NSIT)\cite{mal2016temporal}, the latter condition meaning that measurement does not affect the outcome statistics of any later measurement, analogous to the way the Bell-CHSH inequality was earlier derived from Predictability and No Signaling across \textit{spatial} separation \cite{cavalcanti2012bell}. This feature suggests that if an experiment is set up by choosing the relevant parameters such that the measurement outcomes obtained violate LGI and satisfy the NSIT condition, then these outcomes would be guaranteed to be inherently unpredictable. For quantifying such generated randomness, our treatment will be based on the specifics of the recent experimental test using single photons\cite{joarder2022loophole} that has demonstrated LGI violation by plugging all the relevant loopholes and rigorously satisfying the relevant NSIT conditions. 

The assumptions invoked have been specified with respect to the setup used for the experimental study mentioned earlier, whose  key relevant  features have been discussed in detail in the Appendix of the present paper. Thus, the randomness  certified in this way is  to be regarded as semi-device independent, being  dependent on the extent to which the assumptions invoked have been satisfied.

\textit{The Scheme : }
Consider a single-time evolving system with measurements at various instants of a dichotomic variable $Q$ having eigenvalues $+1$ and $-1$. The Leggett Garg inequality can be written down as, 
\begin{equation}
    \label{Leggett Garg Inequality}
    \braket{Q_1Q_2} + \braket{Q_2Q_3} - \braket{Q_1Q_3} \leq 1
\end{equation}
where $Q_i = Q(t_i)$ is the outcome of the measurement made at time $t_i$ with the flow of time given by, $t_1<t_2<t_3$. The correlation functions are defined as,
\begin{equation}
\label{twotimecorrelators}
    \braket{Q_{i}Q_{j}} = \sum_{a_{i},a_{j} = \pm1}a_{i}a_{j}P(a_i, a_j|Q_i, Q_j)
\end{equation}
where $P(a_i,a_j|Q_i, Q_j)$ is the probability of getting the outcomes $a_i$ and $a_j$ at times $Q_i$ and $Q_j$ respectively. The QM violation of this inequality(with the upper bound of $1.5$) is attributed to the violation of the assumptions characterizing the notion of macrorealism from which LGI is usually derived \cite{leggett1985quantum,emary2013leggett}. However, interestingly, as mentioned earlier, LGI can also be derived from the conjunction of the following assumptions of Predictability and No Signaling in \textit{Time}.
The assumption of Predictability implies that for any given state preparation procedure, all the observable results of measurements at any instant can be uniquely predicted. In this context of a single time-evolving system we are considering, this assumption can be expressed as,
\begin{equation}
\label{predictabilityequation}
    P(a_i, a_j|Q_i, Q_j) \in \{0,1\} .
\end{equation}
The assumption that a measurement cannot affect the observable results of any later measurement is known as the No-Signaling-in-Time condition(also known as the No-Disturbance condition)\cite{kofler2013condition}, which can be expressed as,
\begin{equation}
\label{nsitdefinition}
    P(a_j|Q_j) = \sum_{a_i}P(a_i, a_j|Q_i, Q_j) .
\end{equation}

Relevant to the three-time LGI given by Eq \ref{Leggett Garg Inequality},  the NSIT conditions are as follows
\begin{align}
\label{nosignalingintimeconditions}
    & P(+|Q_2) = P(++|Q_1, Q_2) + P(-+|Q_1,Q_2)\nonumber\\
    & P(+|Q_3) = P(++|Q_1, Q_3) + P(-+|Q_1,Q_3)\nonumber\\
    & P(+|Q_3) = P(++|Q_2, Q_3) + P(-+|Q_2,Q_3) .
\end{align}
%

From this derivation of LGI, it can be argued that in an experimental context where LGI is violated while ensuring the validity of NSIT, the LGI-violating observable outcomes are inherently unpredictable. For obtaining the guaranteed lower bound of the LGI-certified randomness in a semi-device-independent way, we make the following assumptions in the context of our specific experimental setup. First, note that the assumption that the selection of the measurement time is independent of the system’s state, implicit in the derivation of LGI, is satisfied in our setup by ensuring considerable randomness in the choice of the blockers used in the different subsets of runs corresponding to different measurement times. Then the other assumptions invoked in our evaluation of the LGI-certified randomness bound with respect to our setup are listed below: 

\begin{enumerate}
\item The dimension of the system is two. This assumption clearly follows from our setup since the measurements are performed on the spatial degrees of freedom, and there are two paths in the optical setup. Therefore, the state of the photon/system is parametrised using the three parameters $n_x, n_y, n_z$ and can be written down as,
\begin{equation}
    \rho = 1/2(I + \Vec{n}\cdot \Vec{\sigma}) , 
    \ \Vec{n} =(n_x,n_y,n_z) \in \mathbbm{R}^3
\end{equation}
such that $n_x^2+n^2_y+n^2_z \leqslant 1$.

    \item The measurement at times $t_1$ and $t_2$ are the projective measurements defined up-to unitary transformations,
    \begin{equation}
        P_+ = \begin{pmatrix}
    1 & 0 \\
    0 & 0
    \end{pmatrix}, \ 
    P_- =
    \begin{pmatrix}
    0 & 0 \\
    0 & 1
    \end{pmatrix} .
    \end{equation}
    This assumption is sensible here as blockers (pieces of metal) are used for the measurements at $t_1,t_2$. For the measurements at $t_3$ we invoke the general form of $\pm1-$outcome POVM measurement,
    \begin{equation}\label{povmequation}
        M_\pm = \frac12\left((1\pm a)\I \pm \Vec{b}\cdot \Vec{\sigma}  \right), \ \Vec{b} \in \mathbbm{R}^3, \ a \in \mathbbm{R}
    \end{equation} 
    where $|\Vec{b}|\leqslant 1$ and $|\Vec{b}|+|a| \leqslant 1$.
    The measurements at time $t_3$ are carried out by detectors, which are devices with complicated internal workings, unlike the blockers (which are in principle $100 \%$ efficient detectors as has also been characterised in \cite{joarder2022loophole}). Hence, we take the general form of the POVM measurement given by Equation \eqref{povmequation}, which involves an implicit assumption that the blockers do not signal as the POVM at $t_3$ does not depend on the placement of the blockers.
    

\item The initial state is not correlated with any other system thus excluding the possibility of the Eavesdropper having any information about the initial state.

\end{enumerate}

\textit{Bound on Genuine Randomness :} We quantify the randomness generated using the minimum entropy\cite{meng2023maximal, pironio2010random} of the probability distribution, which is defined as,
\begin{align}\label{eq:minimumentropydefinition}
    H_{\infty}(AB|XY) &= -\log\{\text{max}_{a_i, a_j}P(a_i, a_j|Q_i, Q_j)\}\nonumber\\
                      &= -\text{min}_{a_i, a_j}\log\{P(a_i, a_j|Q_i, Q_j)\}.
\end{align}
We now relate the amount of randomness quantified using the minimum entropy to the observed LGI violation. This is done by finding a lower bound on minimum entropy as a function of the LGI violation.
We obtain this bound on minimum entropy by solving the following optimization problem,
\begin{align}\label{subto}
               P^* =     & \ \text{max} \hspace{2 mm}  P(a_i, a_j|Q_i, Q_j)\nonumber \\
                   &\textbf{subject to}\nonumber \\
                   &\braket{Q_1Q_2} + \braket{Q_2Q_3} - \braket{Q_1Q_3} = 1 + \alpha\nonumber\\
                   & P(+|Q_2) = P(++|Q_1, Q_2) + P(-+|Q_1,Q_2)\nonumber\\
                   & P(+|Q_3) = P(++|Q_1, Q_3) + P(-+|Q_1,Q_3)\nonumber\\
                   & P(+|Q_3) = P(++|Q_2, Q_3) + P(-+|Q_2,Q_3)
\end{align}

where $\alpha \in (0,0.5]$. Now the minimal value of the min-entropy, which is compatible with the LGI violation $I$, is given by,
 
\begin{equation}
    H_{\infty}(AB|XY) = -\log_2 P^{*}
\end{equation} \blk
where $P^{*}$ is the solution to the above optimization problem. We derive a bound on minimum entropy as stated in the \textit{Theorem} that follows,
\begin{theorem}\label{thm:rbjp}
Subject to the conditions stated earlier being satisfied, if the three NSIT \eqref{nosignalingintimeconditions} values are zero and the LGI \eqref{Leggett Garg Inequality} value is $1 + \alpha$ where $\alpha \in (0, 0.5]$, then
\begin{equation}
P^* = \frac{1}{4} \left(1 + \alpha + \sqrt{1-2\alpha}\right). \label{eq:pstar}
\end{equation}
Therefore, the guaranteed random bits concerning the amount of violation is given by
\begin{equation}
-\log_2 \left(\frac{1 + \alpha + \sqrt{1-2\alpha}}{4}\right). \label{eq:falpha}
\end{equation}\
\end{theorem}
We briefly outline the proof here, with detailed calculation of the analytical proof of \textit{Theorem 1} and \textit{Theorem 2}, being presented in the Supplementary Material (SM). We use the expressions for the joint probabilities in terms of the parameters defining the unknown state, unitaries, and the measurement at $t_3$ to obtain the expressions for the LGI and NSITs. By suitably utilizing the fact that the NSIT expressions are zero, we establish some relations between the parameters that simplify the LGI expression. 
The problem then simplifies to maximizing the joint probabilities, under the only constraint that the simplified LGI expression is $(1+\alpha)$. We observe that three distinct expressions within the simplified LGI expression are crucial in determining the joint probabilities for the three pairs of measurements. Employing the Lagrange multiplier method, some functional analysis, and intricate mathematical calculations, we identify the maximum values of these three expressions while satisfying the constraint that the simplified LGI value is $(1+\alpha)$.
Consequently, these maximum values help us to compute the upper bounds for all 12 joint probabilities from which we obtain an upper bound on $P^*$. Finally, we present a quantum strategy involving a specific quantum state, unitaries, and measurements that attain this upper bound.

\begin{figure}[]
\centering
\includegraphics[width=0.45\textwidth]{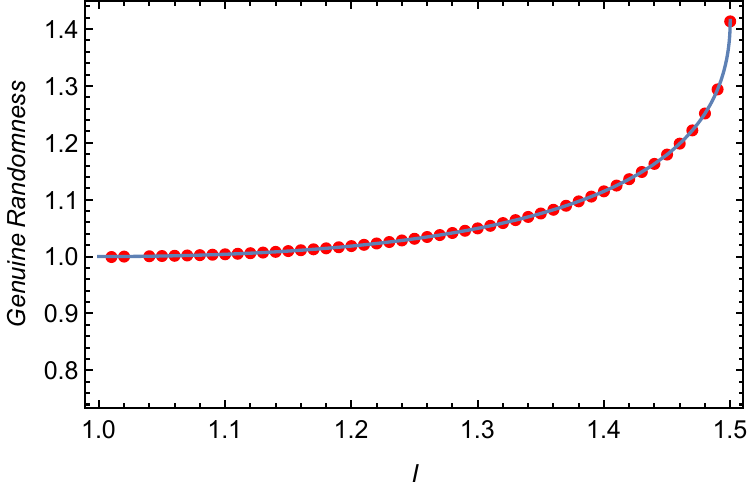}
\caption{Bound on the Genuine Randomness (minimum entropy) for the three-time Leggett Garg setup. This treatment includes the assumption that the system's initial state is not correlated with any other system, thus generating randomness from the joint probabilities $P(a_i,a_j|Q_i, Q_j)$. The blue line is the analytical bound \eqref{eq:pstar} on the minimum entropy, and the red dots are the numerical data from solving the optimization problem. The amount of randomness for the maximal LGI violation is 1.41.} 
\label{fig : BoundonRandomness}
\end{figure}

\textit{Security against state Preparation : }
To ensure security against an adversary, say Eve, accessing initial state information, we adapt our scheme. Firstly, if the user's initial state is entangled with Eve's qubit in a Bell state, Eve can predict the user's measurement outcome by performing her own measurement, compromising security. In this case, the key point is whether we can still ensure an appreciable amount of guaranteed random bits.
Secondly, another possible scenario is when the initial state is a mixture of different pure quantum states fed randomly into each experimental run. Here, the worst-case scenario from a security viewpoint is when Eve can predict the initially prepared state with maximum success.
Even in such a scenario where Eve can maximally guess the outcome of the user's first measurement, we need to ensure that the choice of relevant parameters violates the Leggett-Garg inequality while satisfying all the relevant NSIT conditions, thereby enabling the generation of certified random bits. To achieve the desired security against adversarial attacks, we employ post-processing by quantifying randomness based on user's second measurement outcomes conditioned on first, evaluating guaranteed randomness amount using maximized conditional probability of joint outcome instead of earlier joint probabilities, i.e. evaluating the maximized conditional probability given by $\bar{P^*}$,
 \bea
    \overline{P}^* &=& \max\limits_{\{a_i,a_j,Q_i,Q_j\}}{P(a_j|a_i, Q_i, Q_j)}  \nonumber   \\
  &&  \textbf{subject to constraints in Eq.} \ \eqref{subto} ,
\eea
 
where the mathematical constraints given by Equation \eqref{subto} correspond to violating LGI and satisfying the three relevant NSIT conditions, and the conditional probability is given by,
\begin{equation}\label{condp}
    P(a_j|a_i,Q_i,Q_j) = \frac{P(a_i,a_j|Q_i,Q_j)}{P(a_i|Q_i)} .
\end{equation}
This procedure is based on considering that, for example, in the extreme case of a maximally entangled state shared between Eve and the user, Eve will be able to guess with certainty the outcome of the first $\sigma_z$ measurement by the user using the outcome of her own $\sigma_z$ measurement, which is obviated by the use of conditional probabilities. 
This is possible only when the first measurement is a perfect $\sigma_z$ measurement, which is ensured by the 100 percent efficiency of our blockers.
 The next key question is whether the amount of certified randomness generated by this conditional probability based scheme will be still appreciable, although maybe less than that obtained by the procedure based on joint probabilities discussed earlier. It is this question which is addressed by the following \textit{Theorem 2},
 
 \begin{theorem}\label{thm:rbcp}
Subject to the conditions stated earlier being satisfied, if the three NSIT\eqref{nosignalingintimeconditions} values are zero and the LGI \eqref{Leggett Garg Inequality} value is $1 + \alpha$ where $\alpha \in (0, 0.5]$, then
\begin{equation}\label{cp*}
\overline{P}^* = \frac{1}{2} \left(1 + \alpha + \sqrt{1-2\alpha}\right). 
\end{equation}
Therefore, the amount of guaranteed random bits as a function of $\alpha$ is given by
\begin{equation}\label{lowerboundstatesecure}
f(\alpha) = -\log_2 \left(\frac{1 + \alpha + \sqrt{1-2\alpha}}{2}\right).
\end{equation}
\end{theorem}
 

The proof is essentially an extension of the proof for Theorem \ref{thm:rbjp}, and the relevant details are given in Section IB of SM.
Comparing Equation \eqref{cp*} and \eqref{eq:pstar} it follows that $\bar{P^*} = 2P^*$ and from Equation \eqref{lowerboundstatesecure} it follows that the randomness with respect to the maximum LGI violation (i.e., $\alpha = 1/2$) is 0.415 as compared to 1.41 in the earlier case. Thus an appreciable amount of certified randomness is ensured to be secure against state preparation.
\blk 

\begin{figure}[]
\centering
\includegraphics[width=0.45\textwidth]{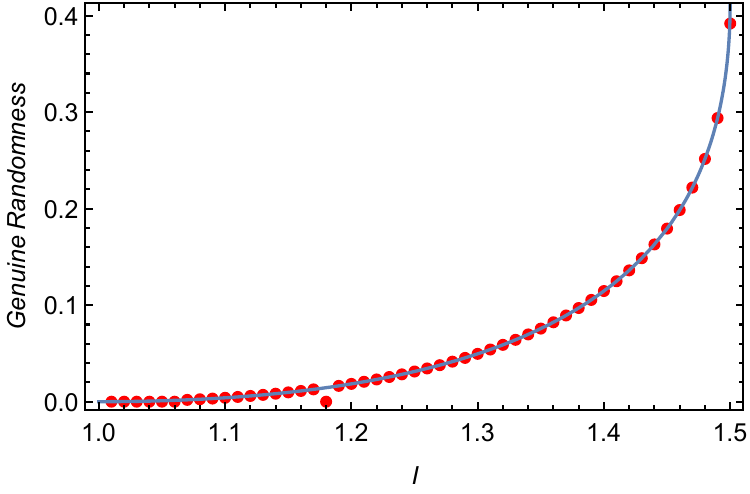}
\caption{Bound on the Genuine Randomness(minimum entropy) for the three-time Leggett Garg setup with full security against state preparation procedure. This is done by solving the optimization problem using the conditional probabilities. The blue line is the analytical bound \eqref{cp*} on the minimum entropy, and the red dots are the numerical data from solving the optimization problem. The amount of randomness for the maximal LGI violation is 0.41, which is, as expected, less than the 1.41 that was earlier obtained assuming secure state preparation. In both these cases, genuine randomness increases monotonically as the LGI violation increases.}
\label{fig:lowerboundstatesecure}
\end{figure}

This bound is sensitive to the NSIT constraint as shown in the SM, where we solve the optimization problem with a small NSIT violation. A higher threshold value of LGI violation is necessary for meaningful randomness generation as NSIT violation becomes more pronounced. Nonetheless, even with a relatively high NSIT violation, a meaningful quantity of random bits can still be obtained as the LGI violation approaches its maximum value.

\textit{Memory Effect and Experimental Results : } To estimate the violation of the LGI, it is necessary to generate data from the device multiple times. However, the device may exhibit variations in performance across different uses, one of the cases being the memory effect, where the output of a particular iteration might depend on the outcome of the previous outputs, hence making it necessary to use a statistical method to account for such memory effects. We have shown in Section II of the SM\cite{supple} how to determine the randomness produced by the devices without making any assumptions about their internal behavior by combining the previously derived bound with a statistical approach.

Due to the memory effect the exact value can be lower than the observed value $\hat{I}$ up to some $\epsilon$, with some small probability $\delta $, 
\begin{equation} \label{8)}
\delta =\exp \left(-\frac{n{\epsilon}^{2} }{2(1/q+I_{q} )^{2} } \right),
\end{equation}
where $I_q$ is the maximum inequality violation allowed by quantum theory, $q = \text{min}\{p(t_1,t_2), p(t_1,t_3), p(t_2,t_3)\}$ and $\epsilon$ is fixed by the maximum LGI violation $I_q$, the probability of the inputs $q$ and the number of runs $n$, as has been defined in Section II of SM.
So the minimum entropy bound of the $n$ bit string generated is,
\begin{equation} \label{memory_effect}
H_\infty(R|S)\ge nf\left(\hat{I}-\epsilon\right)
\end{equation}
with probability at least $1-\delta $. 
With a confidence level of $1-\delta = .99 $ and the experimentally observed LGI violation $I = 1.31$, we have plotted the minimum entropy bound for $n$ runs. In Figure \ref{fig: memoryeffect}, we show that we start getting a substantial amount of randomness only after a certain number of runs due to the presence of the memory effect. Using $n = 10^5$ runs yields a genuine randomness of 3673 bits, corresponding to 0.03673/ bit in the presence of the memory effect. This is lower than expected from the genuine randomness bound derived above, for which we expect a genuine randomness of $0.05406$/bit for an LGI violation of $I = 1.31$.  Moreover, using biased measurement settings increases the threshold for getting an appreciable amount of randomness, as shown in Figure \ref{fig: memoryeffect}.
 \begin{table*}[]
    \label{table : randomnumbers}
    \centering
    \resizebox{\textwidth}{!}{%
    \begin{tabular}{|l|l|l|l|l|}
           \hline
            Performed Experiments & No of Bits & Rate(bits/sec) & Type & Spatial Sep(m) \\
            \hline
             Pironio et al\cite{pironio2010random} & 42 & Not Mentioned & Proof of Concept, Not Loophole free, Uses shielding & 1\\
             \hline 
             P Bierhorst et al\cite{bierhorst2018experimentally}. & 1024 & Not Mentioned & Loophole Free, Randomness Generation & 187 \\
             \hline
             Liu et al \cite{liu2018device} & $6.2469*10^7$ & $181$ &  Randomness Generation & 200 \\
             \hline
             Shen et al \cite{shen2018randomness} & 617,920 & 240 & Randomness Extraction, Assumed No Signaling & Not Mentioned \\
             \hline
             Zhang et al \cite{zhang2020experimental} & 512 & 1.71 & Loophole Free & 194.8 \\
             \hline
             Ming Hang Li et al\cite{li2021experimental} & $5.47*10^8$ & $11598$ & Randomness Expansion & 191\\
             \hline
             Wen Zhao Liu et al\cite{liu2021device} & $2.57*10^7$ & 13,527  & Loophole Free, Randomness Expansion, Uses shielding & Not Mentioned\\
             \hline
             LK Shalm et al\cite{shalm2021device} &  1,181,264,23 & 3606  & Randomness Expansion & 194.8 \\
             \hline 
             \textbf{Our current work } & \textbf{919,118} & \textbf{3865} & \textbf{Loophole free Proof of Concept} & Irrelevant  \\
        
            & & & \textbf{Randomness generation} & \\
             \hline       
    \end{tabular}%
    }
    
       \caption{Comparison of generation rate, type of experiment (proof of concept, loophole-free, and randomness expansion), and the spatial separation of Bell Inequality (BI) based randomness generation experiments with our case of Leggett Garg Inequality (LGI) based randomness generation Experiment. Unlike the BI based experiments, which require spatial separation or some sort of shielding to ensure no signaling, this spatial separation is irrelevant in our case since we can design our experimental setup in a tabletop experiment to ensure NSIT. BI based experiments evolved from proof of concept to loophole-free-experiments, enhancing generation rates and expansion. Our LGI based demonstration, a loophole-free proof of concept experiment, provides the base with an appreciable generation rate. Further improvements and work on expansion schemes for our protocol will boost  LGI based state-of-the-art random number generation.}
    
\end{table*}

\begin{figure}[]
    \centering
    \includegraphics[width=0.45\textwidth]{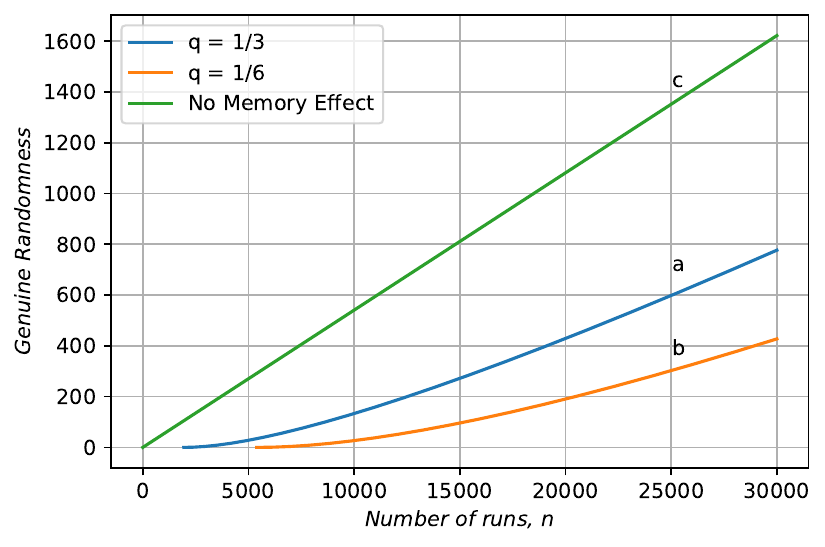}
   
    \caption{In the secure state preparation procedure, we investigate the relationship between genuine randomness and the number of runs in the presence of memory effect. Assuming a violation of the Leggett-Garg inequality with a value of $1.31$ which was observed in our experiment, with a confidence interval of $1-\delta = 0.99$, we see that a notable amount of genuine randomness emerges only after approximately $3000$ runs (curve a) due to the memory effect, compared to the case without memory effect (curve c). For $10^5$ runs, the measured genuine randomness reaches $3673$ with an unbiased seed with probabilities $p(t_1,t_2) = p(t_2,t_3) = p(t_3,t_1) = 1/3$ . Additionally, we investigate the relationship between genuine randomness and the number of runs when using a biased seed, where the measurement settings are chosen with unequal probabilities, $p(t_1,t_2) = 1/6$ and $p(t_2,t_3) = p(t_3,t_1) = 5/12$. While in the unbiased case, non-vanishing randomness starts appearing after $3000$ runs, this threshold increases to around $6000$ runs in the biased case (curve b). For $10^5$ runs in the biased case, the measured genuine randomness reaches $2777$, lower than the unbiased case, as expected.}

    \label{fig: memoryeffect}
\end{figure}

A series of eight experiments were conducted to evaluate various coincidence measurements. Each experiment was repeated multiple times, and the coincidence counts were recorded for $10$ seconds in separate runs. A total of $1,000$ coincidence datasets were collected for each experiment to estimate the LGI violation. The estimated LGI violation from the experiment is $I = 1.32 \pm 0.04$. Considering experimental non-idealities, the corresponding QM prediction is $I_{QM} = 1.34 \pm 0.06$. In addition, another experiment was employed to estimate the single probabilities at times $t_2$ and $t_3$ to verify the NSIT conditions. The experimentally measured values for the three NSIT conditions denoted by $v_1$, $v_2$, and $v_3$ were found to be $0.002 \pm 0.017$, $0.002 \pm 0.016$, and $0.004 \pm 0.016$, respectively.  The QM predictions for these probabilities are $v_1^{QM} = 0$, $v_2^{QM} = 0$, and $v_3^{QM} = 0 \pm 0.0261$. These results certify the randomness of the outputs generated by providing insights into the violations of the LGI and the adherence to the NSIT conditions based on experimental measurements. The average generation rate is  3865 bits/second, and the total number of bits generated is 919,118, as shown in the Appendix. 


    
\textit{Conclusion and Outlook :} Our single system-based RNG scheme's operational advantage over the Bell Inequality based RNGs is that there is no requirement to produce and preserve entanglement across distant systems while measuring randomness-certifying correlations between their observed properties. Fundamentally, there is a key difference in how randomness is certified - entanglement schemes violate Bell inequalities invoking no-signaling across space-like separation, while our scheme certifies randomness through LGI violation invoking no-signaling-in-time, which may not hold in any given experimental configuration. Crucially, our scheme uses setups which satisfy NSIT while violating LGI empirically.

Our treatment provides a fully analytical evaluation of how the lower bound on guaranteed randomness varies monotonically with the LGI violation amount, in complete agreement with correspoding numerical results. While this randomness quantification has operational significance, it can also stimulate a line of studies analogous to the way the nuances of the quantitative relationship between Bell inequality violating randomness and non-locality have been probed in recent years.


Ensuring security against adversary tampering with state preparation is distinct from Bell-based schemes. The most general attack in this scenario is when the user's initial state is  entangled with the adversary's state. To consider the possibility of such an attack, we evaluate guaranteed random bits against the maximized conditional probability of obtaining joint outcomes satisfying no-signaling-in-time conditions and violating LGI. This randomness quantification security strategy is unique to LGI-based schemes and could guide security analysis for other single system quantum randomness generation variants.



 In addition to randomness generation through Bell tests, several interesting semi-device-independent and source-independent schemes have been implemented in diverse experimental setups \cite{avesani2018source,pivoluska2021semi,PRX.10.041048,PRX.6.011020,PRL.114.150501,PRA.94.060301,PhysRevApplied.7.054018}. Additionally, some schemes have been theoretically suggested within sequential measurement setups \cite{Das2022robustcertification,sarkar2023certification}, distinct from our approach. It would be valuable to thoroughly examine and compare the security of these approaches against the potential loopholes. In contrast to the source independent setup we do not make any assumptions about the detectors, which is the main measurement part. Detectors are usually intricate devices with complex internal mechanisms, and thus vulnerable to eavesdropping.

 It is worth mentioning that selecting smaller measurement time intervals without affecting setup-stability can be achieved by automating blocker-position switching using a pseudo-random number generator \cite{shalm2021device}. A thorough examination of randomness expansion in relation to seed randomness could be a potential avenue for future research.

Interestingly,  for counteracting the possible memory effect in the experimental device,  our treatment yields results similar to that for the entanglement-based random generation scheme, requiring a significant number of runs to generate a substantial amount of certified randomness. A more rigorous estimation of the amount of randomness  considering into account the possible side information available to the adversary and the relevant generation rate by employing randomness extraction and amplification will be presented in future work, along with studies investigating the possibility of other variants of this scheme in terms of experimental setups showing the violation of LGI using different systems. 

\textit{Acknowledgements : } U.S. acknowledges partial support provided by the Ministry of Electronics and Information Technology (MeitY), Government of India under a grant for Centre for Excellence in Quantum Technologies with Ref. No. 4(7)/2020-ITEA as well as partial support from the QuEST-DST Project Q-97 of the Government of India.  We also thank Aninda Sinha for useful discussions.

\appendix
\newpage
\clearpage
\section{Appendix}
\label{experimentandresults}

We provide thorough details of our Experimental Setup for LGI violation, addressing all loopholes and meeting NSIT requirements to ensure suitability for randomness generation. Additionally, we outline the process of generating random bits from this Experimental Setup.

\textit{Experimental Setup} : The experimental setup of Ref \cite{joarder2022loophole} we are considering for generating LGI-certified randomness consists of three stages,
\begin{enumerate}
    \item State Preparation: This step used a single photon source and a beam splitter to generate a pair of photons, out of which one is sent for heralding and the other is sent to the experimental setup. 
    \item Unitary Transformation: The two unitary transformations($t_1 \rightarrow t_2$ and $t_2\rightarrow t_3$) were implemented using an Asymmetric Mach-Zender Interferometer(AMZI) and a displaced Sagnac interferometer(DSI).
    \item Measurements: Measurements were performed using blockers in different arms of the two interferometers for noninvasive measurements (NIM) and single-photon avalanche detectors (SPAD) for direct detection at the end of the experiment.
\end{enumerate}
 
\textit{State Preparation}: A heralded twin-photon source was built based on spontaneous parametric down-conversion (SPDC), with a diodelaser pumping a BBO crystal with a $405$ nm wavelength and $10$ mW power. The BBO crystal is oriented so that it is phase-matched for degenerate, non-collinear, type-I SPDC while being pumped with horizontally polarized light. Parametric down-conversion creates pairs of single photons with vertical polarization and 810 nm central wavelength. To increase pair generation, we also place a focusing lens (L1) to focus the pump beam into the central spot of the BBO crystal. A long-pass filter (F1) is placed after the crystal to block the pump beam and pass only the down-converted single-photon pairs. A half-wave plate and polarising beam splitter PBS1 are placed after the non-linear crystal to separate the two photons in the two arms of the beam splitter. Two mirrors are placed to direct one photon to the experiment and the other to a SPAD1 detector for heralding.

\begin{figure}[]
\centering
\includegraphics[width=0.45\textwidth]{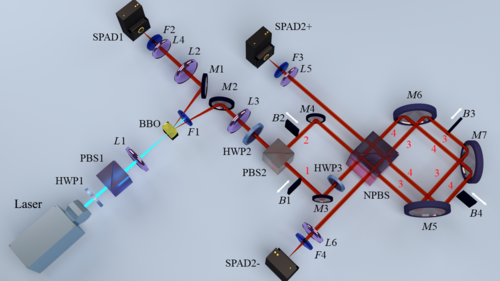}
\caption{Schematic of the experimental setup. Here HWP1, HWP2, and HWP3 are the half-wave plates; PBS1 and PBS2 are the polarizing beam splitters; L1, L4, L5, and L6 are the focusing lens; F1 is the long-pass filter; M is the dielectric mirror; L2 and L3 are the collimating lenses; F2, F3, and F4 are the band-pass filters; B1, B2, B3, and B4 are the blockers; NPBS is the nonpolarizing beam splitter; and SPAD1, SPAD2+, and SPAD2- are the single-photon avalanche detectors. Two arms of the AMZI are marked as 1 and 2, representing the +1 and -1 arms, respectively. Similarly, two arms of the DSI are marked as 3 and 4, representing -1 and +1. SPAD2+ and SPAD2- are placed in the +1 and -1 arms, respectively. Adapted with permission from Joarder et al., 2022, PRX Quantum {\bf 3} 010307, 2022 \cite{joarder2022loophole}.}
\label{fig:experimentdiagram}
\end{figure}

\textit{Unitary Transformation:} The experimental setup consists of two interferometers whose arms are denoted by 1,2,3,4, where blockers are placed for noninvasive measurements. The first interferometer is an asymmetric Mach-Zehnder interferometer (AMZI), while the second is a displaced Sagnac interferometer (DSI). The beam-splitting ratio in the two arms of the AMZI is controlled by a combination of a half-wave plate (HWP2) and a polarizing beam splitter (PBS2).  For satisfying the two-time NSITs,  the two arms of the first Mach-Zehnder interferometer (MZI) are made noninterfering by adding a path difference between the $+1$ and $-1$ arms. A single nonpolarizing beam splitter (NPBS) with a measured splitting ratio of 80:20 (concerning vertically polarized light at 810 nm wavelength) is used in the DSI. Two detectors (SPAD2+ and SPAD2- ) are placed in the two output arms of the DSI to detect single photons.

The time $t_1$, $t_2$ and $t_3$ are being defined in the following manner:
\begin{itemize}
    \item $t_1$ is the time from PBS2 to the first impact on NPBS
    \item $t_2$ is the time from the first impact to the second impact on NPBS
    \item $t_3$ is the time after the impact on NPBS till detection on one of the detectors.
\end{itemize}

\textit{Measurements}: Negative result measurements at $t_1$ and $t_2$ are performed using motorized blockers (B1 and B2) in arms 1 and 2 and (B3 and B4) in arms 3 and 4. The experiment is completed in three stages corresponding to the measurement of $\langle Q_{t_1}Q_{t_3} \rangle$, $\langle Q_{t_2}Q_{t_3} \rangle$ and $\langle Q_{t_1}Q_{t_2} \rangle$ respectively. For the first two stages, two runs each are performed by placing the blockers on the respective arms and detecting the photon at the end to measure the coincidence events $(++)$, $(+-)$, $(-+)$, and $(--)$. For instance, if a blocker is placed in the $-$ arm of the second interferometer(DSI), and a click is observed in SPAD2+, this will count as a measurement for the probability $P(++|Q_2Q_3)$ and a click in SPAD2- will count as a measurement for the probability $P(+-|Q_2Q_3)$. For the third stage, i.e., for the measurement of $\langle Q_{t_1}Q_{t_2} \rangle$, four runs are performed to evaluate the three-time probabilities. For example,  when blockers are placed in the - arm of AMZI and in the - arm of DSI, a detection in SPAD2+ will count as $P(+++|Q_1Q_2Q_3)$, and a detection in SPAD2- will count as $P(++-|Q_1Q_2Q_3)$. These probabilities are then marginalized to evaluate the two-term probabilities at time $t_1$ and $t_2$, which leads to $\langle Q_{t_1}Q_{t_2} \rangle$. $P(+|Q_3)$ was computed by conducting the experiment without any blockers and $P(+|Q_2)$ was computed by placing a blocker at the negative arm of the second interferometer and marginalizing the two time probabilities .
Only the coincidence counts measured, i.e., the simultaneous detection of SPAD1 and SPAD2+ or SPAD2- are considered valid counts in evaluating the probabilities.
We have used avalanced photo diode detectors which have inherently a  reasonably higher dark count.  A follow up experiment could change this to superconducting nanowire based detectors, which have higher quantum efficiency as well as lower dark counts. This in turn will affect the signal to noise ratio of the results and can lead to higher rate of random bit generation.

\textit{Addressing Loopholes:} To ensure the experiment was loophole-free, various measures were taken. The clumsiness loophole was addressed using non-invasive measurements (NIM) and tuning the experimental parameters to satisfy the two-time NSIT conditions. The detection efficiency loophole was eliminated by showing that the violation of LGI cannot be reproduced by the hidden variable model, regardless of detection efficiency. The pivotal aspect of our setup, wherein the measurement at $t_3$ is consistently performed for all the choices of measurement times, plays a crucial role in overcoming this loophole \cite{joarder2022loophole}.  The multi-photon emission loophole was addressed using a heralded single-photon source and appropriate filtering. The coincidence loophole was eliminated by using a pair of photons as a timing reference and adjusting the coincidence time windows accordingly. Finally, the preparation state loophole was closed by post-selecting only those detected photons from the SPDC source and choosing high signal-to-noise ratios for the corresponding coincidence time windows.

 \textit{Random number generation} : From the eight experiments conducted, we selected three datasets from each experiment to generate bit strings composed of `0's and `1's. The generation of random numbers was based on the coincidence clicks of two detectors, SPAD2+ and SPAD2-, with the heralding detector SPAD1. Coincidence counts were identified using information from the heralding detector and employing a $4ns$ time window. We designated detecting a coincidence event at SPAD2+ as `0' and detecting a coincidence event at SPAD2- as `1'.

For the evaluation of the probabilities $P(a_i, a_j|Q_1, Q_3)$ and $P(a_i, a_j|Q_2, Q_3)$ in the first and second phases of the experiment, two sub-runs were conducted for each experiment. In one sub-run, the + arm of the first interferometer was blocked, and in the other sub-run, the - arm of the interferometer was blocked. In the first case, if a photon from the experimental setup coincidentally hit SPAD2+ with the heralding detector SPAD1, it was counted as `0'. If it coincidentally hit SPAD2- with SPAD1, it was counted as `1', thus generating a bit string for this sub-run and resulting in the probabilities $P(-+|Q_1, Q_3)$ and $P(--|Q_1, Q_3)$. Similarly, for the second sub-run where the - arm was blocked, a bit string was generated based on the detector clicks, leading to the probabilities $P(++|Q_1, Q_3)$ and $P(+-|Q_1, Q_3)$.

Likewise, two more bit strings were generated from the second phase of the experiment, providing the probabilities $P(a_i,a_j|Q_2, Q_3)$. However, the third phase of the experiment, aimed at computing correlations at times $t_1$ and $t_2$, involved marginalizing the three-time probabilities $P(a_i, a_j, a_k|Q_1, Q_2, Q_3)$. In this case, blockers were placed simultaneously on both interferometers in different arms, enabling the computation of all the three-term probabilities in 4 runs.

 For example, when both + arms of the interferometers were blocked, the detector counts yielded bit strings corresponding to the three-term probabilities $P(--+|Q_1, Q_2, Q_3)$ and $P(--+|Q_1, Q_2, Q_3)$. Although these bit strings did not directly originate from the two-term probabilities $P(a_i, a_j|Q_1, Q_2)$, which occur in the LGI expression used for certifying randomness, they eventually contributed to the computation of two-term probabilities. They thus could be used to certify and quantify the randomness. 
 

Subject to the conditions assumed in this approach, eight distinct bit strings can be generated, as shown in Table II, using the available data from the experiments focused on coincidence event calculations. The average generation rate is 3865 bits/second, and the total number of bits generated, which is the sum of the 8-bit strings generated, is 919,118. Each bit string had an appropriate length and successfully passed the SP-800-90B entropy test\cite{rukhin2001statistical}\cite{turan2018recommendation} for randomness.

\vspace{2 mm}
\begin{table}[H]
    \label{table : randomnumbers}
    \centering
    \begin{tabular}{|l|l|l|}
           \hline
            Experiment & Rate(bits/sec)) & Length \\
            \hline
             P($-$$-$|23) P($-$$+$|23) & 4722 & 140382 \\
            \hline
             P($+$$-$|23) P($+$$+$|23) & 5139 & 152405\\
            \hline
             P($+$$-$$-$|123) P($+$$-$$+$|123) & 1177 & 34981\\
            \hline
             P($+$$+$$-$|123) P($+$$+$$+$|123) & 4268 & 127123\\
            \hline
             P($-$$-$$-$|123) P($-$$-$$+$|123) & 3953 & 117651\\
            \hline
             P($-$$+$$-$|123) P($-$$+$$+$|123) & 1180 & 34935 \\
            \hline
             P($+$$-$|13) P($+$$+$|13) & 5158 & 153465 \\
            \hline
             P($-$$-$|13) P($-$$+$|13) & 5321 & 158176\\
            \hline
            
    \end{tabular}
    \caption{Length of the random bit string generated from the detector counts of the two detectors SPAD2+ and SPAD2- from the 8 experiments to evaluate the different joint probabilities.}
    
\end{table}
\appendix
\newpage
\clearpage

\onecolumngrid
{\begin{center}\bf \Large{Supplementary material}\end{center} }

\section{Bounds on Genuine Randomness}


A general two-dimensional quantum state can be parameterized as
\be 
\rho = \frac12 (\I + \Vec{n}\cdot \Vec{\sigma}) , \ \Vec{n} =(n_x,n_y,n_z) \in \mathbbm{R}^3
\ee 
such that $n_x^2+n^2_y+n^2_z \leqslant 1$. We take the general form of unitaries $U_1,U_2$ as 
\be 
U_i = 
\begin{pmatrix}
e^{\im x_i} \cos[z_i] & e^{\im y_i} \sin [z_i] \\
-e^{- \im y_i} \sin [z_i] & e^{-\im x_i} \cos [z_i]
\end{pmatrix} \text{ for } i=1,2 , \quad x_i,y_i,z_i \in \mathbbm{R} .
\ee 
Without loss of generality, we take the measurement at $t_1$ and $t_2$ to be diagonal defined by the following projectors,
\be 
P_+ = \begin{pmatrix}
1 & 0 \\
0 & 0
\end{pmatrix}, \ 
P_- =
\begin{pmatrix}
0 & 0 \\
0 & 1
\end{pmatrix} .
\ee 
Moreover, the most general form of the measurement at $t_3$ is defined by two positive operators $M_+,M_- \geqslant 0$ such that $M_++M_-=\I$, which can be expressed as
\be 
M_\pm = \frac12\left((1\pm a)\I \pm \Vec{b}\cdot \Vec{\sigma}  \right), \ \Vec{b} \in \mathbbm{R}^3, \ a \in \mathbbm{R}
\ee 
where $|\Vec{b}|\leqslant 1$ and $|\Vec{b}|+|a| \leqslant 1$.
Without loss of generality, we can consider $VM_\pm V^\dagger$ for any unitary $V$ by absorbing $V$ into $U_2$. Thus, we can take $M_\pm$ to be diagonal as follows
\be 
M_\pm = \frac12\left((1 \pm a)\I \pm b \sigma_z  \right), \ a\in \mathbbm{R} , \ b \in \mathbbm{R}^+ ,
\ee 
and 
\be \label{abc}
|a|+b \leqslant 1, \ b \leqslant 1 .
\ee 

\subsection{Without Security against state Preparation}
Using the state, unitaries, and measurements, we compute the following expression of the joint probabilities,
\bea 
P(++|Q_1,Q_2) = \tr ( U_1P_+ \rho P_+ U_1^\dagger P_+ ) = \frac12 (1 + n_z) \cos[z_1]^2 \label{p12pp}\\
P(+-|Q_1,Q_2) = \tr ( U_1P_+ \rho P_+ U_1^\dagger P_- ) = \frac12 (1 + n_z) \sin[z_1]^2 \label{p12pm} \\
P(-+|Q_1,Q_2) = \tr ( U_1P_- \rho P_- U_1^\dagger P_+ ) = \frac12 (1 - n_z) \sin[z_1]^2 \label{p12mp}\\
P(--|Q_1,Q_2) = \tr ( U_1P_- \rho P_- U_1^\dagger P_- ) = \frac12 (1 - n_z) \cos[z_1]^2 \label{p12mm}
\eea 
It follows from the above four equations that
\be \label{p12c}
\langle Q_1Q_2 \rangle = \cos[2 z_1] .
\ee 
Similarly, the other joint probabilities can be obtained,
\bea
& P(++|Q_1,Q_3) = \tr ( U_2U_1P_+ \rho P_+ U_1^\dagger U_2^\dagger M_+ ) = \frac14 (1 + n_z) (1 + a + \gamma ) \label{p13pp}\\
& P(+-|Q_1,Q_3) = \tr ( U_2U_1P_+ \rho P_+ U_1^\dagger U_2^\dagger M_- ) = \frac14 (1 + n_z) (1 - a - \gamma ) \label{p13pm} \\
& P(-+|Q_1,Q_3) = \tr ( U_2U_1P_- \rho P_- U_1^\dagger U_2^\dagger M_+ ) = \frac14 (1 - n_z) (1 + a - \gamma) \label{p13mp} \\
& P(--|Q_1,Q_3) = \tr ( U_2U_1P_- \rho P_- U_1^\dagger U_2^\dagger M_- ) = \frac18 (1 - n_z) (1 - a + \gamma) \label{p13mm}
\eea 
where 
\be \label{gamma}
\gamma = b (\cos[2 z_1] \cos[2 z_2] -  \cos[t] \sin[2 z_1] \sin[2 z_2] )
\ee 
and
\be \label{t}
t = x_1 + x_2 + y_1 - y_2 .
\ee 
 Therefore,
\be \label{p13c}
\langle Q_1Q_3 \rangle = a n_z + \gamma ;
\ee 
and
\bea 
\label{p23pp}
 P(++|Q_2,Q_3) = \tr ( U_2 P_+ U_1 \rho U_1^\dagger P_+ U_2^\dagger M_+ ) = \frac14 (1 + a + b \cos[2 z_2]) (1 + n_z \cos[2 z_1] + \chi ) \\
 \label{p23pm}
P(+-|Q_2,Q_3) = \tr ( U_2 P_+ U_1 \rho U_1^\dagger P_+ U_2^\dagger M_- ) = \frac14 (1 - a - b \cos[2 z_2]) (1 + n_z \cos[2 z_1] +  \chi ) \\
 \label{p23mp}
P(-+|Q_2,Q_3) = \tr ( U_2 P_- U_1 \rho U_1^\dagger P_- U_2^\dagger M_+ ) = \frac14 (1 + a - b \cos[2 z_2]) (1 - n_z \cos[2 z_1] -  \chi ) \\
\label{p23mm}
    P(--|Q_2,Q_3) = \tr ( U_2 P_- U_1 \rho U_1^\dagger P_- U_2^\dagger M_- ) = \frac14 (1 - a + b \cos[2 z_2]) (1 - 
     n_z \cos[2 z_1] - \chi ) ,
\eea 
and thus,
\be \label{p23c}
\langle Q_2Q_3 \rangle = a n_z \cos[2 z_1] + b \cos[2 z_2] + 
 a \chi .
\ee 
Using \eqref{p12c}, \eqref{p23c}, \eqref{p13c}, we get the expression for the LGI correlator as,
\be \label{lgi}
\text{LGI} = (1+ a n_z)\cos[2 z_1] + b \cos[2 z_2] - a n_z - \gamma .
\ee 
and the three NSIT conditions of Eq 5 in the main text as,
\be \label{nd1}
\text{NSIT}_1 = P(+|Q_2) - P(++|Q_1, Q_2) - P(-+|Q_1,Q_2) = \frac12 \chi ,
\ee 
where 
\be \label{chi}
\chi = (n_x \cos[x_1 - y_1] + n_y \sin[x_1 - y_1]) \sin[2 z_1] ;
\ee 
\be \label{nd2}
\text{NSIT}_2 =
     P(+|Q_3) - P(++|Q_1, Q_3) - P(-+|Q_1,Q_3) =
     \frac12 b (\cos[2 z_2] \chi  + \sin[2 z_2] \xi ) ,
\ee 
where 
\be \label{xi}
\xi = \cos[t] \cos[2 z_1] (n_x \cos[x_1 - y_1] + n_y \sin[x_1 - y_1]) + \sin[t] (n_y \cos[x_1 - y_1] -  n_x \sin[x_1 - y_1]) ;
\ee 
and
\be \label{nd3}
\text{NSIT}_3 = P(+|Q_3) - P(++|Q_2, Q_3) - P(-+|Q_2,Q_3) = \frac12 b  \sin[2z_2] (\xi - n_z \cos[t] \sin[2z_1]) .
\ee  
First, we employ the feature that the three NSITs are zero to simplify the optimization. It is immediate that NSIT$_1=0$ implies $\chi =0$. Substituting $\chi=0$ into \eqref{nd2} and using the fact that NSIT$_2=0$, we find
\be 
b \sin[2z_2]\xi = 0.
\ee 
Substituting this relation into the condition \eqref{nd3}, we obtain
\be \label{nsit30}
b n_z \cos[t] \sin[2z_1]  \sin[2z_2] =0 .
\ee 
If $b=0$, the LGI expression \eqref{lgi} becomes $(1-n_z)\cos[2z_1] + n_z$, which is clearly less than or equal to 1. Thus, we arrive at a contradiction that we do not have any violation. If 
\be 
\cos[t] \sin[2z_1]  \sin[2z_2] = 0 ,
\ee 
then the LGI expression \eqref{lgi} reduces to
\be \label{intlgi}
(1+ a n_z)\cos[2 z_1] + b \cos[2 z_2] - a n_z -  b \cos[2 z_1] \cos[2 z_2] .
\ee 
Due to the following argument, the above quantity \eqref{intlgi} is also less than or equal to 1 for any values of $z_1,z_2,b,a,n_z$.
By equating the partial derivative of \eqref{intlgi} with respect $n_z$ to 0, we get either $a=0$ or $\cos[z_1]=1$ for the maximum value of this expression. The first case simplifies the expression \eqref{intlgi} to $\cos[2 z_1] + b \cos[2 z_2] -  b \cos[2 z_1] \cos[2 z_2]$, which is less than 1 since $b\leqslant 1$. For the second case, the expression becomes 1. So, it leads to a contradiction with LGI violation, and thus, \eqref{nsit30} must imply $n_z = 0$. Altogether, the fact that the three NSITs are zero implies
\be \label{nsitsim}
\chi = n_z = 0 .
\ee 
By replacing $n_z=0$ into the LGI expression \eqref{lgi} and taking the LGI value to be $(1+\alpha)$, one arrives at the following relation 
\be \label{lgiob}
\text{LGI} = \cos[2 z_1] + b \cos[2 z_2] - b \cos[2 z_1] \cos[2 z_2] + b \cos[t] \sin[2 z_1] \sin[2 z_2]  = 1+\alpha ,
\ee 
wherein $b\neq 0$ and  $ \alpha \in (0,0.5] $.
The next step is to obtain $P^*$ from the above relation. To do so, we take the help of the following lemma.

\begin{lemma}
Suppose we have two variables $x,y$ such that $x,y\in (-1,+1)$ satisfying the constraint 
\be \label{op}
x+by-bxy + kb \sqrt{1-x^2}\sqrt{1-y^2} = 1+\alpha 
\ee 
where $\alpha \in (0,0.5], b\in (0,1],$ and $k\in [-1,1]$. Then the maximum value of $x$ is $\alpha+\sqrt{1-2\alpha}$; the maximum value of $by$ is $\alpha+\sqrt{b^2-2\alpha}$; and the maximum value of the expression $(-b-bxy+kb\sqrt{1-x^2}\sqrt{1-y^2})$ is $\alpha+\sqrt{1-2\alpha} -1$. 
\end{lemma}
\begin{proof}
We redefine the constraint \eqref{op} as 
\be \label{gxy}
G(x,y)= x+by-bxy + kb \sqrt{1-x^2}\sqrt{1-y^2} - 1 -\alpha = 0 .
\ee 
For obtaining the maximum and minimum value of $x$, we take the assistance of the Lagrange multiplier method
\be 
\nabla_y x = \lambda \nabla_y G(x,y) ,
\ee 
which, after some steps, leads to the following relation
\be 
y = \sqrt{\frac{1-x}{k^2(1+x)+(1-x)} } .
\ee 
Replacing this expression of $y$ into \eqref{op} and after some simplifications, we arrive at a quadratic equation of $x$,
\be 
(k^2b^2 - b^2 +1 )x^2 -2(1+\alpha-b^2) x + (1+\alpha)^2 - b^2(1+k^2) = 0,
\ee
the solution of which is given by
\be 
x  = \frac{1}{1+b^2k^2 - b^2} \left(1+\alpha -b^2 \pm b \sqrt{\alpha^2 - (2\alpha+\alpha^2)k^2 + b^2 k^4} \right) .
\ee 
It can be verified that within the range of values of $k^2 \in [0,1], b\in (0,1]$ such that $x \in (-1,+1)$, the above larger solution (with the + sign) of $x$ is increasing with $k^2$ and $b$ since the derivatives are positive in that range. Thus, the maximum value is obtained for $k^2=b=1$; consequently, the maximum value of $x$ is $\alpha + \sqrt{1-2\alpha}$.

We follow a similar method to obtain the maximum value of $yb$. With the aid of $\nabla_x by = \lambda \nabla_x G(x,y)$, we first get 
\be 
 x = \frac{1-by}{\sqrt{k^2b^2(1-y^2) + (1-by)^2} } .
\ee 
Replacing this expression of $x$ in \eqref{op} leads us to the following quadratic equation of $y$ after some simplifications
\be 
k^2(by)^2 - 2\alpha (by) + \alpha^2 + 2\alpha - k^2b^2 = 0 .
\ee 
Taking $by$ as the variable, the solution of the above is
\be 
by = \frac{1}{k^2} \left( \alpha \pm \sqrt{\alpha^2 - (\alpha^2 + 2\alpha) k^2 + b^2k^4} \right) ,
\ee 
which is again maximum for $k^2=1$ with + sign whenever $by \in (-1,+1)$. Thus, the maximum value of $by$ is $\alpha + \sqrt{b^2-2\alpha}$.

To find the maximum value of $(-b-bxy+kb\sqrt{1-x^2}\sqrt{1-y^2})$, 
we first get the following relation by equating $\lambda$ from the two Lagrange equations $\nabla_x (-b-bxy+kb\sqrt{1-x^2}\sqrt{1-y^2})= \lambda \nabla_x G(x,y) $ and $\nabla_y (-b-bxy+kb\sqrt{1-x^2}\sqrt{1-y^2}) = \lambda \nabla_y G(x,y)$, 
\be \label{bxyeq1}
k( y - x^2y -bx + bxy^2 ) = (by - x) \sqrt{1-x^2}\sqrt{1-y^2} . 
\ee 
Let us note that the expression $(-b-bxy+kb\sqrt{1-x^2}\sqrt{1-y^2})$ remains invariant if we interchange the variables $x$ and $y$. Thus, if the maximum value of this expression is obtained for some $x=x^*,y=x^*$, then $x=y^*,y=x^*$ also yields its maximum value. Therefore, the following equation with the interchange between $x$ and $y$ in \eqref{bxyeq1} should also hold,
\be \label{bxyeq2}
k( x - y^2x -by + byx^2 ) = (bx - y) \sqrt{1-x^2}\sqrt{1-y^2} .
\ee 
From \eqref{bxyeq1} and \eqref{bxyeq2}, we get another relation,
\be 
(x-y^2x-by+byx^2)(by-x) = (y-x^2y-bx+bxy^2)(bx-y) ,
\ee 
after using the facts that $x\neq \pm1$, $y\neq \pm1$, and $k\neq 0$ since  $\alpha >0 $. A straightforward calculation shows that the above equation implies, either $b=0$ or $x=\pm y$. We know that $b\neq 0$, otherwise the right-hand-side of \eqref{op} cannot be greater than 1. If we replace $x=-y$ in \eqref{bxyeq1}, we will get $k=1$. Subsequently, by substituting $x=-y,k=1$ into \eqref{op}, one finds $x = 1+\alpha/(1-b)$ which is always greater than 1. So, this cannot be a correct solution since $x\in (-1,1)$. By replacing the only remaining option, $x=y$, into \eqref{bxyeq1}, we get either $k=-1$ or $x=0,\pm 1$ or $b=1$. Clearly, $x$ cannot be $0$ or $\pm 1$. If $k=-1$ and $x=y$, then \eqref{op} suggests that the value of $x=1+\alpha/(1-b)$, which is also greater than 1. Hence, we discard this option. As a consequence, we must have $b=1$. Finally, substituting $x=y$ and $b=1$ into \eqref{op}, we arrive at  
\be \label{simopx=y}
(1+k)x^2 - 2x + 1+\alpha - k = 0 .
\ee 
The solution of this quadratic equation of $x$ is given by
\be 
x = \frac{1}{1+k} \left( 1 \pm \sqrt{1 - (1+k)(1+\alpha-k)} \right) .
\ee 
The minimum value of the above expression is 
\be \label{minxbxy}
\frac12 \left( 1- \sqrt{1-2\alpha} \right) ,
\ee 
when $k=1$ and the sign is negative.   
On the other hand, for $x=y$ and $b=1$, the expression,
\bea  
-b-bxy+kb\sqrt{1-x^2}\sqrt{1-y^2} &=& -(1+k)x^2 + k - 1 \nonumber \\
& = & \alpha  - 2x \nonumber \\
& \leqslant & \alpha + \sqrt{1-2\alpha} -1 ,
\eea  
where the second line is obtained using \eqref{simopx=y}, and the third is obtained by restoring the minimum value of $x$ from \eqref{minxbxy}.
\end{proof}  

We can identify the variables $x,y,k$ from \eqref{op} by $\cos[2z_1],\cos[2z_2],\cos[t]$ in \eqref{lgiob}, respectively. By using the above lemma, the reduced form of LGI value \eqref{lgiob} implies
\be \label{cosz1}
\cos[2z_1] \leqslant \alpha + \sqrt{1-2\alpha},
 \ee 
 \be \label{cosz2}
 b \cos[2z_2] \leqslant \alpha + \sqrt{b^2-2\alpha} ,
 \ee 
 and 
 \be \label{cosz1z2}
- b- b\cos[2 z_1] \cos[2 z_2] + b \cos[t] \sin[2 z_1] \sin[2 z_2] \leqslant  \alpha + \sqrt{1-2\alpha} -1 .
 \ee 
Note that, due to \eqref{cosz2}, $b\geqslant \sqrt{2\alpha}$. Moreover, for the maximum violation, i.e., $\alpha=1/2$, $b=1$ signifying the measurement at $t_3$ to be projective.
Let us now evaluate the values of the joint probabilities. Putting $n_z=0$ in \eqref{p12pp} and \eqref{p12mm}, and applying \eqref{cosz1} we get
 \bea \label{p12b}
 P(++|Q_1,Q_2) = P(--|Q_1,Q_2)
 & = & \frac14 (\cos[2z_1]+1) \nonumber \\
 &\leqslant &  \frac{1}{4}\left( 1+ \alpha + \sqrt{1-2\alpha} \right) .
 \eea 
The other probabilities $P(+-|Q_1,Q_2),P(-+|Q_1,Q_2)$ must be less than this value as the sum of all four is 1. Substituting \eqref{nsitsim} into the joint probabilities pertaining to $Q_2,Q_3$, we find
 \bea 
  P(\pm \pm |Q_2,Q_3)  
& = & \frac14 \left(1 \pm a + b \cos[2 z_2] \right) \nonumber \\
 & \leqslant &  \frac14 \left(2 - b + b\cos[2 z_2] \right) \nonumber  \\
 & \leqslant &  \frac14 \left(2 - b + \alpha + \sqrt{b^2-2\alpha}  \right) .
 \eea 
The second line is obtained by using $|a|\leqslant 1-b$ from \eqref{abc}. The third line is due to \eqref{cosz2}. The derivative of the expression $\left(2 - b + \alpha + \sqrt{b^2-2\alpha}  \right)$ is
$b/(\sqrt{b^2-2\alpha}) - 1 ,$
which is always positive within the interval $\alpha \in (0,0.5]$ and $b\in [\sqrt{2\alpha},1]$. Therefore, it is increasing with $b$ within the interval $[\sqrt{2\alpha},1]$, and therefore, the maximum is achieved at $b=1$. Subsequently, we have
\be  \label{p23b}
  P(\pm \pm|Q_2,Q_3) \leqslant \frac{1}{4}\left( 1+ \alpha + \sqrt{1-2\alpha} \right) .
\ee 
The other two probabilities must be less than this value. Due to \eqref{nsitsim}, \eqref{abc}, \eqref{gamma},  the probabilities \eqref{p13pm}-\eqref{p13mp} simplify to
\bea \label{p13b}
P(\mp \pm |Q_1,Q_3)  
 &=& \frac14 \left(1 \pm a  - \gamma \right) \nonumber \\
 & \leqslant &  \frac14 \left(2 - b - b \cos[2 z_1] \cos[2 z_2] + b \cos[t] \sin[2 z_1] \sin[2 z_2] \right) \nonumber \\
  & \leqslant &   \frac14 \left( 1+ \alpha + \sqrt{1-2\alpha} \right) ,
\eea 
where the last line is found using \eqref{cosz1z2}.
%
The relations \eqref{p12b},\eqref{p23b},\eqref{p13b} altogether imply
\be 
P^* \leqslant  \frac14 \left( 1+ \alpha + \sqrt{1-2\alpha} \right) .
\ee 
Finally, in order to show that this upper bound is tight, that is, this upper bound is the exact value, it suffices to provide a quantum strategy that achieves this value. By performing a numerical optimization we came up with the quantum state, unitaries, and measurements defined by the parameter values,
\bea \label{qust}
& a=n_x=n_y=n_y =0, \nonumber \\
& \cos[t]=b=1, \nonumber \\
& \cos[2z_1] = \cos[2z_2] = (1-\sqrt{1-2\alpha})/2 . 
\eea 
which satisfies the three NSIT expressions \eqref{nd1},\eqref{nd2},\eqref{nd3} and gives LGI value \eqref{lgiob} $1+\alpha$. Using this we can calculate the probability, 
\be 
P(+ - |Q_1,Q_3) = \frac14 \left( 1+ \alpha + \sqrt{1-2\alpha} \right) . 
\ee 
which shows that the bound is tight.

\textit{Numerical Estimation : } Using the expressions for LGI and NSIT in terms of the parameters for the states unitaries and the generalized measurement, we numerically solve the optimization problem stated in the main text using the optimization tools of Mathematica which matches with the analytical bound derived above.

\subsection{Security against state Preparation}


Let us first find the denominators of the conditional probabilities of Equation 14 in the main text. From \eqref{p12pp}-\eqref{p12mm}, \eqref{p13pp}-\eqref{p13mm}, and \eqref{p23pp}-\eqref{p23mm}, we find the following expressions,
\bea 
& P(\pm|Q_1) = \frac12(1\pm n_z), \nonumber \\
& P(\pm|Q_2) = \frac12(1 \pm n_z \cos[2z_1] \pm \chi ).
\eea 
Substituting \eqref{nsitsim}, that is, whenever the three NSIT conditions are satisfied, the probabilities reduce to $P(\pm|Q_i) = 1/2$ for $i=1,2$. Consequently, each conditional probability is two times the respective joint probability. This implies that the desired quantity $\overline{P}^* = 2 P^*$, and we obtain the bound using the result of \textit{Theorem} 1.

\blk 


\section{Relaxing NSIT constraint} 
In our analysis, we have explored the ideal scenario where the No-Signaling-In-Time (NSIT) condition is fully satisfied along with LGI violation, leading to a completely random output and a violation of predictability. 
However, it is important to acknowledge that real-world experiments do not always satisfy the NSIT conditions and are satisfied up to a certain tolerance. In light of this, we have derived a bound that ensures a minimum level of assured randomness even in the cases for which NSIT is satisfied up to a certain tolerance, giving us a deeper understanding of the intricate interplay between the extent of NSIT satisfaction and the preservation of the minimum level of certified randomness.

We will solve the following optimization problem to numerically evaluate the minimum entropy bound when NSIT is not satisfied,

\begin{align*}
    &P^{*}(a_j|a_i,Q_i,Q_j) =  \text{max} \hspace{2 mm}  P(a_j|a_i,Q_i,Q_j) \\
                   &\textbf{subject to} \\
                   &\braket{Q_1Q_2} + \braket{Q_2Q_3} - \braket{Q_1Q_3} = 1 + \alpha \\
                   & P(+|Q_2) -P(++|Q_1, Q_2) - P(-+|Q_1, Q_2) = v\\
                   & P(+|Q_3) - P(++|Q_1, Q_3) - P(-+|Q_1,Q_3) = v\\
                   & P(+|Q_3) - P(++|Q_2, Q_3) - P(-+|Q_2,Q_3) = v\\
\end{align*}

 The result of the above optimization problem, as shown in Figure \ref{fig:randomnessvsnsitviolation}, indicates that as the violation of the No Signaling in Time (NSIT) conditions increases, the ability to generate high-quantity randomness decreases. Additionally, as the violation of NSIT becomes more pronounced, a higher threshold value of Leggett-Garg inequality (LGI) violation is needed to generate substantial randomness.  But even with a relatively high NSIT violation, a meaningful amount of random bits can still be obtained as the LGI violation approaches its maximum value. Here, we have shown how Genuine Randomness varies for some particular values of NSIT. However, we note that this trend is currently restricted to this assumption, and while there is some indication that there is some functional relationship, it calls for deeper studies that involve increasing the parameter space in the same sense as was done for studies involving probing the relationship between Bell inequality violations, genuine randomness and Non locality\cite{acin2012randomness}.

\begin{figure}[H]
\centering
\includegraphics[width=0.45\textwidth]{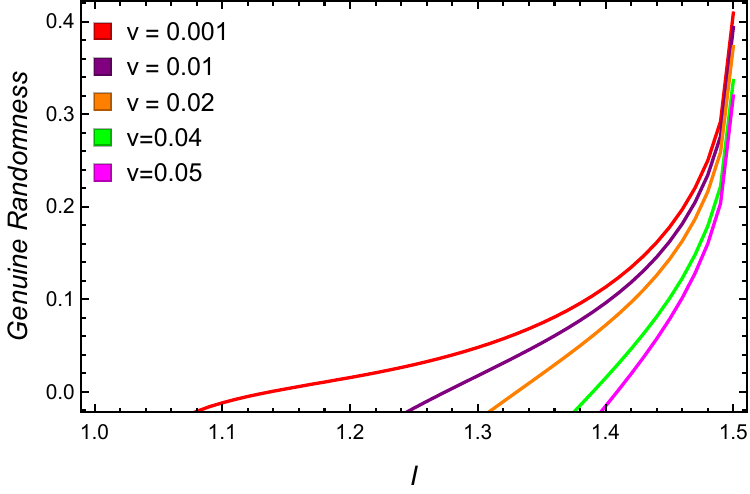}
\caption{Variation of genuine randomness as a function of NSIT violation v. As the violation of NSIT increases, the threshold value of the violation of LGI required to generate an appreciable amount of randomness also increases. But still, even up to the violation of NSIT being 0.05, the amount of random bits produced remains significant while approaching the maximum violation of LGI. }
\label{fig:randomnessvsnsitviolation}
\end{figure}

\section{Memory Effect for Conditional Probabilities}
\label{app: memoryeffectconditionalprobabilities}
To estimate the violation of the Leggett-Garg Inequality, it is necessary to generate data from the device multiple times. However, the device may exhibit variations in performance across different uses, one of the cases being the memory effect, where the output of a particular iteration might depend on the outcome of the previous outputs, hence making it necessary to use a statistical method to account for such memory effects\cite{barrett2002quantum}. We will demonstrate how to determine the randomness produced by the devices without making any assumptions about their internal behavior by combining the previously derived bound with a statistical approach.

Suppose we use the devices repeatedly $n$ times. Let $x_{i} ,y_{i} \in \{ 1,\ldots ,m\} $ be the inputs and $a_{i} ,b_{i} \in \{ 1,\ldots ,d\} $ be the outputs for each round $i$. We define $a^{k} =(a_{1} ,a_{2} ,\ldots ,a_{k} )$ as the first $k$ outputs $a_{i} $, similarly for $b^{k} $, $x^{k} $, and $y^{k} $. The input pairs ($x_{i},y_{i} $) at each round are random variables with the same distribution $P(x_{i}=x,y_{i}=y )=P(x,y)$, but $P(x,y)$ may not be a product distribution $P(x,y)=P(x)P(y)$.

Let $P_{R|S} =\left\{P(b^{n}  | a^{n} x^{n} y^{n} )\right\}$ be the conditional probability distribution of the final output string $r=(b^{n} )$ given the fact that the sequence of inputs $s=(x^{n},y^{n} )$ has been inserted in the devices and the string of initial output bits is $(a^{n})$

The min-entropy can now characterize the randomness of the output string conditioned on the inputs, 

$$H_\infty(R|S)=\min _{b^{n}}\left(-\log _{2} P( b^{n} | a^{n} x^{n} y^{n} )\right)$$

Now if we minimize $-\log _{2} P( b^{n} |a^{n} x^{n} y^{n} )$ wrt $(a^{n}, b^{n} $) and $(x^{n} ,y^{n} )$ then we can derive  a lower bound on $H_\infty(R|S)$, as $$H_\infty(R|S)\ge -\log _{2} P^{*}( b^{n} |a^{n} x^{n} y^{n} )$$

Now the conditional probability can be written down as,

\begin{eqnarray} \label{eq1}
-\log_2 P(a^nb^n|x^ny^n) &=&-\log_2\prod_{i=1}^n P(a_ib_i|x_iy_i) \nonumber\\
&=&-\log_2\prod_{i=1}^n \frac{P(b_i|a_ix_iy_i)}{P(a_i|x_i)}\nonumber\\
&=&- \log_2 \frac{P(b^n|a^nx^ny^n)}{P(a^n|x^n)}
\end{eqnarray}

If the events were independent, then the combined probability can be written down as a product of the individual runs,

\begin{equation}
\label{eq2}
    -\log_2 P(a^nb^n|x^ny^n) = -\log_2\prod_{i=1}^n P(a_ib_i|x_iy_i)
\end{equation}

Similarly, the combined probability for only the first measurement can be given by,

\begin{equation}\label{eq3}
    -\log_2 P(a^n|x^n) = -\log_2\prod_{i=1}^n P(a_i|x_i)    
\end{equation}

But we assume that the result of the $i^{th}$ trial depends on the results of all the $(i-1)^{th}$ runs so that the probability can be written down as the product of all the probabilities conditioned to the previous inputs and outputs. Moreover we assume that the output at round $i$ does not depend on future inputs $(x_j,y_j)$ with $j>i$
\begin{eqnarray} \label{eq4}
-\log_2 P(a^nb^n|x^ny^n)&=&-\log_2\prod_{i=1}^n P(a_ib_i|a^{i-1}b^{i-1}x^iy^i)\nonumber\\
&=&-\log_2\prod_{i=1}^n P(a_ib_i|x_iy_iW^i)\\
&=&\sum_{i=1}^n - \log_2 P(a_ib_i|x_iy_iW^i)\nonumber
\end{eqnarray}

The variable $W^{i} =(a^{i-1} b^{i-1} x^{i-1} y^{i-1} )$ is used to denote all events in the past of round $i$. 

Similarly, for the single measurement probabilities we have,

\begin{eqnarray} \label{eq5}
-\log_2 P(a^n|x^n)&=&-\log_2\prod_{i=1}^n P(a_i|a^{i-1}b^{i-1}y^{i-1}x^i)\nonumber\\
&=&-\log_2\prod_{i=1}^n P(a_i|x_iW^i)\\
&=&\sum_{i=1}^n - \log_2 P(a_i|x_iW^i)\nonumber
\end{eqnarray}

Now from equation \eqref{eq1}, \eqref{eq4} and \eqref{eq5} we have,

\begin{eqnarray}\label{eq6}
    -\log_{2}P(b^n|a^nx^ny^n)= -\sum_ {i=1}^n \log_2 P(b_i|a_ix_iy_iW^i)
\end{eqnarray}
The behavior of the devices at round $i$ conditioned on the past is characterized by a response function $P(a_{i} b_{i} |x_{i} y_{i} W^{i} )$ and an LGI violation $I(W^{i} )$. 

Now we have derived a bound on the probabilities for each trial,  

$$-\log _{2} P( b_{i} |a_{i}x_{i} y_{i} )\ge f(I)$$

Whatever the precise form of the quantum state and measurements implementing this behavior, they are bound to satisfy the constraint 
$$-\log _{2} P( b_{i} |a_{i}x_{i} y_{i} W^{i} )\ge f(I(W^{i} ))$$

Now we can insert this relation into Eq. \eqref{eq6},

\begin{equation}\label{eq7}
-\log_2 P(b^n|a^nx^ny^n) \geq \sum_{i=1}^n f\left(I(W^i)\right)]
\end{equation}

We have derived the bound on the probabilities for the Leggett Garg Inequality, and it takes the form, 

\begin{equation}\label{eq8}
    f(I) = -\log_2[\frac{1+\alpha + \sqrt{1-2\alpha}}{4}]
\end{equation}

where $\alpha \in (0,0.5]$. Now since this function is convex, we can write the above inequality as, 

\bea 
\label{eq9}
-\log_2 P(b^n|a^nx^ny^n) \geq&nf\left(\frac{1}{n}\sum_{i=1}^nI(W^i)\right)
\eea 

Now we will show a way of evaluating the quantity $\frac{1}{n}\sum_{i=1}^nI(W^i)$ in \eqref{eq9}, which can be estimated from the experimental data. This can be done in three steps:
\vspace{2 mm}

\textit{Step1 : Define an Estimator}

\vspace{2mm}

First, we will define a quantity that uses the output data $a$, $b$, and the measurement settings $x$ and $y$ to estimate the $LGI$ violation. Let us define a random variable,

\begin{equation}
    \label{random variable}
    \hat{I}_i = \sum_{abxy} c_{abxy} \frac{\chi(a_i = a , b_i = b, x_i = x, y_i = t)}{P(x,y)}
\end{equation}

 The random variable is defined in such a way that the expectation on the past $W^i$ is $E(\hat{I_i}|W^i) = I(W^i)$. The quantity $\chi(e)$ for an event e is $1$ if the event has occurred and is 0 if the event hasn't occurred.  The sum of the random variable for the $n$ iterations of the experiment, $\hat{I} = \sum_{i=1}^n {\hat{I}}_i$, estimates the $LGI$ violation for the experiment. We can show this by using the appropriate coefficients $c_{abxy}$, such that \eqref{random variable} corresponds to the LGI expression given in Equation 1 in the main text.

Let  $q = min_{xy}{P(x,y)}$ be the minimum probability of the measurement settings that we use, and we assume that $q>0$. 

\vspace{2 mm}
\textit{Step 2: Construct a sequence and prove it is a martingale}

Now in order to approximate the quantity $\frac{1}{n}\sum_{i=1}^nI(W^i)$ with the estimator that we defined above, we will have to construct a martingale out of these quantities and apply bounds on martingale increment. To do that, let us consider the sequence, 

\begin{equation}
    \label{sequence}
    Z^k = \sum_{i=1}^{k}(I_i - I(W_i))
\end{equation}

Now in order for the sequence $\{Z^k: k\geq1\}$ to be a martingale with respect to the sequence $\{W^k: k\geq2\}$ we will have to verify the following two properties of martingale,

\begin{enumerate}
    \item $E(|Z^k|)\leq \infty$
    \item $E(Z^k|W^1,W^2, \dots W^j) = E(Z^k|W^j) = Z^j$
\end{enumerate}

$I_i$ has a maximum value of $1/q$, and $I(W^i)$ is bounded by the maximum possible violation of the LGI inequality allowed by quantum mechanics, which we can denote by $I_q$. Since we assume $q\neq 0$ and $I_q$ is finite, therefore from the triangle inequality, the sequence $Z^k$ is bounded, implying that the expectation value is also bounded. 

From the definition of $W$, $W^k$ contains all the information of the $W^j$ where $j\leq k$, implying $E(Z^k|W^1,W^2, \dots W^j) = E(Z^k|W^j)$.

\begin{align*}
    E(Z^k|W^j) &= E(Z^j|W^j) + E(\sum_{i=j+1}^k(I_
    i - I(W^i))|W^j)\\
    &=Z^j\\
\end{align*}

Hence $\{Z^k: k\geq1\}$ to be a martingale with respect to the sequence $\{W^k: k\geq2\}$.

\vspace{2 mm}
\textit{Step3 : Bound on martingale}
\vspace{2 mm}

As a final step will use the Azuma-Hoeffding inequality\cite{tong2014possible,lalley2013concentration}, which is given by the following theorem.

\textbf{Theorem}: Let $S_n$ be a martingale relative to some sequence $Y_n$ satisfying $S_0 = 0$ and whose increments, 

$\zeta_n = s_n - s_n-1$ are bounded by $\mod{\zeta_k} \leq \sigma_k$ then, 

$$P(S_n\geq \alpha) \leq \exp\left(\frac{-\alpha^2}{\sum_{j=1}^n\sigma_j^{2}}\right)$$

Now, from the triangle inequality, 
\begin{align*}
Z^{k} - Z^{k-1} &= \mod{I_k - I(W^k)}\\
                & \leq \mod{I_k} + \mod{I(W^k)}\\
                & \leq \frac{1}{q} + I_q\\
\end{align*}

Hence taking $\alpha = n\epsilon$ the Azuma Hoeffding inequality implies, 

\begin{equation}
    P(Z^n \geq n \epsilon) \leq \exp\left(\frac{-n\epsilon^2}{2(1/q + I_q)^2}\right)
\end{equation}

Using this bound we can say that the quantity $\frac{1}{n} \sum _{i=1}^{n}I (W^{i} )$ can be lower than the observed value $\hat{I}=\frac{1}{n} \sum _{i=1}^{n}I_{i}  $ up to some $\epsilon$ only with some small probability $\delta $, 

 \begin{equation} \label{eq10}
\delta =\exp \left(-\frac{n{\epsilon}^{2} }{2(1/q+I_{q} )^{2} } \right).
\end{equation}

Combining this last result with Eq. \eqref{eq5}, we conclude that
\begin{equation} \label{eq11}
H_\infty(R|S)\ge -\log _{2} P(a^{n} b^{n} |x^{n} y^{n} )\ge nf\left(\hat{I}-\epsilon\right)
\end{equation}
with probability at least $1-\delta $.

\section{NSIT Conditions under memory effect }
\label{app : third}
We want to study how the memory effect is altering the NSIT conditions. The NSIT conditions that are being used in our protocol are given below, 

\begin{align}
\label{NSIT Equations}
    & P(+|Q_2) - P(++|Q_1, Q_2) - P(-+,Q_1,Q_2) = 0\nonumber\\
    & P(+|Q_3) - P(++|Q_1, Q_3) - P(-+,Q_1,Q_3) = 0\nonumber\\
    & P(+|Q_3) - P(++|Q_2, Q_3) - P(-+,Q_2,Q_3) = 0
\end{align}

We will use a similar treatment as used in Section II where we assume that the behavior of the devices at round $i$ conditioned on the past inputs and outputs is characterized by a response function $P(a_ib_i|x_iy_iW_i)$ and an NSIT violation of $N^j(W^i)$ where $j = 1,2,3$ and  $W^i = (a^{i-1}b^{i-1}x^{i-1}y^{i-1})$ denotes all the events in the past of round $i$. 

We will use a similar indicator function for the NSITs,

\begin{align}
\label{indicatorNSIT}
\hat{N_i^j} = &\sum_{m_{abxy}^j} m_{abxy}^j \chi(a_i = a, b_i = b, x_i = x, y_i = y)\\
\end{align}
where $\chi(e)$ is the indicator function of the event $e$ $i.e$ $\chi(e) = 1$ if the event has occurred and $\chi(e) = 0$ if the event does not occur. $a_i$ and $b_i$ denote measurement outcomes at round $i$, and $x_i$ ,$y_i \in  \{0, 1, 2, 3\}$ denote the measurement settings, where $0$ indicates no measurement. Now we define $\hat{N}^j = \frac{1}{n} \sum{\hat{N}_i^j}$ where the label $j$ indicates the particular NSIT condition. We can show that with the proper choice of coefficients $m$, $\hat{N}^j$ gives us the three NSIT conditions.

Now let us introduce the random variables, $Z^k_j$ for $j = 1, 2, 3$
\begin{equation}
\label{randomvariableNSIT}
Z^k_j= \sum_{i=1}^k {N_i^j - N^j(W^i)}
\end{equation}

With similar calculations as in Section II, we can show that each of these $Z_k's$ are martingales with respect to some sequence $W_k$. Now the range of martingale increment is bounded by, 

\begin{equation}
\label{martingaleincrement}
    |N^j_i - N^j(W^i)| \leq 1 + N_q^j
 \end{equation}
where $N_q^j$ is the maximum violation of the NSIT conditions allowed by quantum theory. 

So from the Azuma-Hoeffding Inequality, we can show that due to the memory effect, the NSIT will differ from the value obtained from the experiment by an amount $\epsilon$ with a probability $\delta$, 

\begin{equation}
    \label{AzumaHoeffdingNSIT}
    P\left(\frac{1}{n}\sum_{i=1}^nN_i^j - \frac{1}{n} \sum_{i=1}^n{\hat{N}^j(W^i)} \leq \epsilon_j\right)\leq \delta_j
\end{equation}

where 
\begin{equation}
    \delta_j = \exp\left(-\frac{n \epsilon_j^2}{2(1+N_q^j)}\right)
\end{equation}

In the context of the three NSITs, where the quantum bound for each NSIT is $N_q^j = 1/2$, we can demonstrate the impact of the memory effect on the estimated NSIT values. By considering a high confidence interval of 
$1-\delta =.99$, we observe that the deviation between the experimentally measured values and the NSIT values, accounting for the memory effect, approaches zero as the number of runs, $n$, increases. Specifically, when the number of runs is approximate $n \approx O(10^5)$, the deviation between the estimated NSIT values and the experimental values is on the order of $\epsilon \approx O(10^{-2})$. This finding indicates that, for large values of $n$, the presence of a memory effect does not significantly impact adherence to the NSIT conditions.

\begin{figure}[h]
    \centering
    \includegraphics{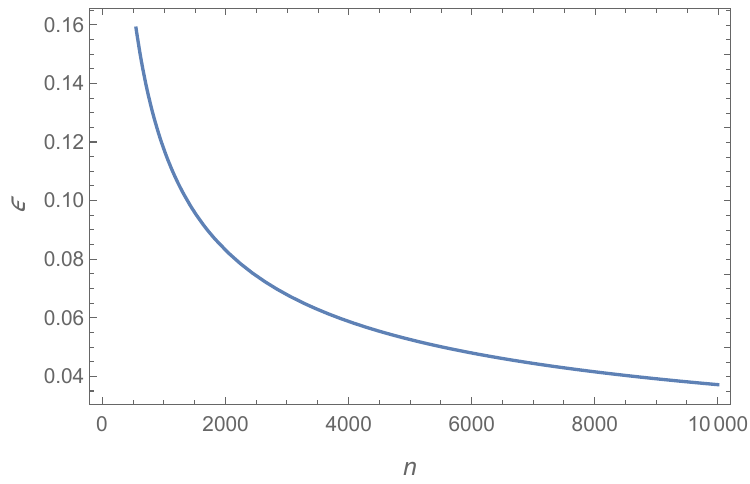}
    \caption{As the number of runs increases, the deviation in the value of the No Signaling in time relations due to the memory effect from the experimental value decreases. When the number of runs is of the order $n \approx O(10^5)$ the deviation is of the order $\epsilon \approx O(10^{-2})$ for a confidence interval of $1 - \delta = 0.99$}
    \label{fig:enter-label}
\end{figure}

\end{document}